\documentclass{amsart}[12pt]
\usepackage{amsmath}
\usepackage{graphicx}
\usepackage[titletoc,toc]{appendix}
\usepackage{latexsym}
\usepackage{hyperref}
\usepackage{graphicx}
\usepackage{amsthm,amssymb,amstext,amscd,amsfonts,amsbsy,amsxtra,latexsym}
\usepackage{wrapfig}
\usepackage{color}
\usepackage{pdflscape}
\usepackage{lscape}
\usepackage{tabularx}
\usepackage{lipsum}
\usepackage{rotating} 
\usepackage{pifont}
\usepackage{upgreek}
\usepackage{mathrsfs}
\usepackage{wasysym}
\usepackage[top=2.5cm, bottom=2.5cm, left=2.5cm, right=2.5cm]{geometry}

\newcommand{\lb}{\label}
\newcommand{\bea}{\begin{eqnarray}}
\newcommand{\eea}{\end{eqnarray}}

\makeatletter

\renewcommand{\tocsubsection}[3]{%
  \indentlabel{\@ifnotempty{#2}{\hspace*{2.3em}\makebox[2.3em][l]{%
    \ignorespaces#1 #2.\hfill}}}#3}
\renewcommand{\tocsubsubsection}[3]{%
  \indentlabel{\@ifnotempty{#2}{\hspace*{4.6em}\makebox[3em][l]{%
    \ignorespaces#1 #2.\hfill}}}#3}
		
\newcounter{mnotecount}[section]

\renewcommand{\themnotecount}{\thesection.\arabic{mnotecount}}

\newcommand{\mnote}[1]
{\protect{\stepcounter{mnotecount}}$^{\mbox{\footnotesize $%
\!\!\!\!\!\!\,\bullet$\themnotecount}}$ \marginpar{
\raggedright\tiny\em $\!\!\!\!\!\!\,\bullet$\themnotecount: #1} }

\theoremstyle{plain}
\newtheorem{theorem}{Theorem}[section]
\newtheorem{proposition}[theorem]{Proposition}
\newtheorem{lemma}[theorem]{Lemma}

\theoremstyle{definition}

\newtheorem{remark}[theorem]{Remark}

\numberwithin{equation}{section}

\allowdisplaybreaks[3]

\swapnumbers
\begin{document}

\begin{center}

\title[]{The wave equation near flat Friedmann-Lema\^itre-Robertson-Walker and Kasner Big Bang singularities}

\author{}

\address{}

\email{}

\date{}

\parskip = 0 pt

\maketitle

\bigskip
Artur Alho\footnote{e-mail address: aalho@math.ist.utl.pt}{${}^{,\star}$},
Grigorios Fournodavlos\footnote{e-mail address: gf313@cam.ac.uk}{${}^{,\dagger}$}, and
Anne T. Franzen\footnote{e-mail address: anne.franzen@tecnico.ulisboa.pt}{${}^{,\star}$} 

\bigskip
{\it {${}^\star$}Center for Mathematical Analysis, Geometry and Dynamical Systems,}\\
{\it Instituto Superior T\'ecnico, Universidade de Lisboa,}\\
{\it Av. Rovisco Pais, 1049-001 Lisboa, Portugal}

\bigskip
{\it {${}^\dagger$} Department of Pure Mathematics and Mathematical Statistics,\\
University of Cambridge,\\
Wilberforce Road,
Cambridge
CB3 0WB, United Kingdom}

\end{center}

\begin{abstract}
We consider the wave equation, $\square_g\psi=0$, in fixed flat Friedmann-Lema\^itre-Robertson-Walker 
and Kasner spacetimes with topology $\mathbb{R}_+\times\mathbb{T}^3$. We obtain generic blow up results for solutions to the wave equation towards the Big Bang singularity in both backgrounds. In particular, we characterize open sets of initial data prescribed at a spacelike hypersurface close to the singularity, which give rise to solutions that blow up in an open set of the Big Bang hypersurface $\{t=0\}$. The initial data sets are characterized by the condition that the Neumann data should dominate, in an appropriate $L^2$-sense, up to two spatial derivatives of the Dirichlet data. For these initial configurations, the $L^2(\mathbb{T}^3)$ norms of the solutions blow up towards the Big Bang hypersurfaces of FLRW and Kasner with inverse polynomial and logarithmic rates respectively. Our method is based on deriving suitably weighted energy estimates in physical space. No symmetries of solutions are assumed.
\end{abstract}

\section{Introduction and main theorems}

In this note, we analyse the behaviour of solutions to the wave equation on cosmological backgrounds towards the initial singularity. 
Our spacetimes of interest are the spatially homogenous, isotropic flat  Friedmann-Lema\^itre-Robertson-Walker (FLRW hereafter) backgrounds and the anisotropic vacuum Kasner spacetimes. The former plays an important role in physics, since observational evidence 
suggests that at sufficiently large scales the universe seems to be spatially homogeneous and isotropic. 
The Kasner solutions also play an important role in the theory of general relativity, since they form the past attractor of Bianchi type I 
spacetimes, the {\it Kasner circle}, which in turn are the basic building blocks in the ``BKL conjecture"~\cite{BKL70} concerning spacelike cosmological singularities, see \cite{Brehm} for recent developments on this subject in the setting of spatially homogeneous solutions to the Einstein-vacuum equations. 

The spacetimes have the topology $\mathbb{R}_+\times\mathbb{T}^3$ and are endowed with the metrics:
\begin{align}
\label{FLRWmetric}g_{\text{FLRW}}=&-dt^2+t^{\frac{4}{3\gamma}}(dx_1^2+dx_2^2+dx^2_3),\qquad\frac{2}{3}<\gamma<2,\\
\label{Kasnermetric}g_{\text{Kasner}}=&-dt^2+\sum^{3}_{j=1}t^{2p_j}dx_j^2,\qquad\sum^{3}_{j=1}p_j = 1,\;\;  \sum^{3}_{j=1}p^2_j = 1,\;\; p_j<1,
\end{align}
respectively. Both metrics $g_{\mathrm{FLRW}},g_{\mathrm{Kasner}}$ have a Big Bang singularity at $t=0$, where the curvature blows up $|\text{Riem}|\sim t^{-2}$, as $t\rightarrow0$. 
Metrics of the form (\ref{FLRWmetric}) are solutions of the Einstein-Euler system for ideal fluids with linear equation of state $p=(\gamma-1)\rho$, 
where $p$ is the pressure, and $\rho$ the energy density. The case $\gamma=2$ corresponds to stiff fluids, i.e., $p=\rho$, 
where incompressibility is expressed by the velocity of sound $c_s$ equating the velocity of light $c=1$. 
For the stiff case the dynamics of the Einstein equations towards the singularity are 
completely understood by the work of \cite{RS1,RS2,Speck}. The other endpoint $\gamma=\frac{2}{3}$ corresponds to the coasting universe which does not have a spacelike singularity.
On the other hand, the Kasner metric (\ref{Kasnermetric}) is a solution to the Einstein vacuum equations. When one of the $p_j=1$ equals one, and the other two vanish (flat Kasner), it corresponds to the Taub form of Minkowski space and there is no singularity as the spacetime is flat.

Our goal is to understand the behaviour of smooth solutions to the wave equation towards these singularities from the initial value problem point of view and by deriving appropriate energy estimates in physical space, which may also prove useful for dynamical studies. According to the references in the literature, such as \cite{AR,peterson,rendall,Rin17}, these waves are shown to blow up in certain cases. We wish to characterize open sets of initial data at a given time $t_0>0$ for which such blow up behaviour occurs at $t=0$. Denote the constant $t$ hypersurfaces by $\Sigma_t$.

First, we give the general asymptotic profile of all solutions:
\begin{theorem}\label{Thm:Asym}
Let $\psi$ be a smooth solution to the wave equation, $\square_g\psi=0$, for either of the metrics $g_{\mathrm{FLRW}}, g_{\mathrm{Kasner}}$, arising from initial data $(\psi_0,\partial_t\psi_0)$ on $\Sigma_{t_0}$. Then, $\psi$ can be written in the following form: 
\begin{align}
\label{psiasymFLRW}\psi_{\mathrm{FLRW}}(t,x)=A_{\mathrm{FLRW}}(x)t^{1-\frac{2}{\gamma}}+u_{\mathrm{FLRW}}(t,x),\\
\label{psiasymKasner}\psi_{\mathrm{Kasner}}(t,x)=A_{\mathrm{Kasner}}(x)\log t+u_{\mathrm{Kasner}}(t,x),
\end{align}
where $A(x),u(t,x)$ are smooth functions and $u_{FLRW}t^{\frac{2}{\gamma}-1},u_{Kasner}(\log t)^{-1}$ tend to zero, as $t\rightarrow0$.
\end{theorem}
We prove the preceding theorem by deriving appropriate stability estimates for renormalized variables, which as a corollary imply the continuous dependence of $A(x)$ on initial data. For instance, solutions coming from initial configurations close to those of the homogeneous solutions $t^{1-\frac{2}{\gamma}},\log t$ in FLRW and Kasner respectively, will blow up with leading order coefficients $A(x)\sim1$. Hence, the set of all blowing up solutions to the wave equation is open and dense.\footnote{In the sense that if a solution blows up in a compact set at $\Sigma_0$, i.e., $A(x)\neq0$ in that compact set, then this property persists under sufficiently small perturbations. On the contrary, if $A(x)=0$ in an open subset of $\Sigma_0$, one can always add a small multiple of $t^{1-\frac{2}{\gamma}},\log t$ to produce a new solution with $A(x)\neq0$ in that open set.} 

Our next theorem yields a characterization of open sets of initial data for which the corresponding solutions to the wave equation blow up in $L^2(\mathbb{T}^3)$ at the Big Bang hypersurface $t=0$.
\begin{theorem}\label{Thm:BlowUP}
Let $\psi$ be a smooth solution to the wave equation, $\square_g\psi=0$, for either of the metrics $g_{\mathrm{FLRW}}, g_{\mathrm{Kasner}}$, arising from initial data $(\psi_0,\partial_t\psi_0)$ on $\Sigma_{t_0}$, $t_0>0$. If $\partial_t\psi_0$ is non-zero in $L^2(\mathbb{T}^3)$, $t_0$ is sufficiently small such that 
\begin{align}
\label{smallt0FLRW}\frac{2t_0^{2-\frac{4}{3\gamma}}}{1-(\frac{2}{3\gamma})^2}\sum_{i=1}^3\|\partial_t\partial_{x_i}\psi_0\|_{L^2(\mathbb{T}^3)}^2<&\,\epsilon\|\partial_t\psi_0\|_{L^2(\mathbb{T}^3)}^2, &\mathrm{(FLRW)}\\
\label{smallt0Kasner}\sum_{i=1}^3 \frac{2t^{2-2p_i}_0}{(1-p_i)^2}\|\partial_t\partial_{x_i}\psi_0\|^2_{L^2(\mathbb{T}^3)}<&\,\epsilon\| \partial_t\psi_0\|^2_{L^2(\mathbb{T}^3)},& (\mathrm{Kasner})
\end{align}
and $\psi_0,\partial_t\psi_0$ satisfy the open conditions
\begin{align}
\label{initcondFLRW}(1-\epsilon)\|\partial_t\psi_0 \|^2_{L^2(\mathbb{T}^3)} >&\, 
t^{-\frac{4}{3\gamma}}_0\sum^3_{i=3}\|\partial_{x_i}\psi_0\|^2_{L^2(\mathbb{T}^3)} +\frac{2t_0^{2-\frac{8}{3\gamma}}}{1-(\frac{2}{3\gamma})^2}\sum_{i,j=1}^3\|\partial_{x_j}\partial_{x_i}\psi_0\|_{L^2(\mathbb{T}^3)},&\mathrm{(FLRW)}\\
\label{initcondKasner} (1-\epsilon)\|\partial_t\psi_0\|^2_{L^2(\mathbb{T}^3)}>&\,
\sum^3_{i=3}t^{-2p_i}_0\|\partial_{x_i}\psi_0\|^2_{L^2(\mathbb{T}^3)} +\sum_{i,j=1}^3\frac{2t_0^{2-2p_i-2p_j}}{(1-p_i)^2}\|\partial_{x_j}\partial_{x_i}\psi_0\|_{L^2(\mathbb{T}^3)},&\mathrm{(Kasner)}
\end{align}
for some $0<\epsilon<1$, then $\|A(x)\|_{L^2(\mathbb{T}^3)}>0$.
\end{theorem}
\begin{remark}
Given a blowing up solution to the wave equation in either FLRW or Kasner, having non-vanishing leading order coefficient $A(x)$, it is easy to see, using the expansions (\ref{psiasymFLRW}), (\ref{psiasymKasner}), that the solution satisfies the conditions in Theorem \ref{Thm:BlowUP} for $t_0>0$ sufficiently small. 
\end{remark}

We also prove a local version of Theorem \ref{Thm:BlowUP}, giving open initial conditions in a neighbourhood of $\Sigma_{t_0}$, $U_{t_0}$, whose domain of dependence intersects the singular hypersurface $\Sigma_0$ at a neighbourhood $U_0$, where the $L^2(U_0)$ norm of the corresponding solutions blows up.
\begin{figure}[h]
	\centering
	\def\svgwidth{9cm}
	\includegraphics[scale=1.2]{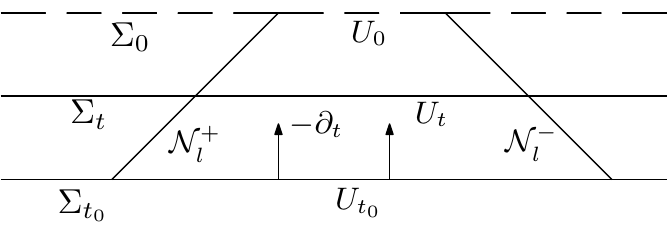}
	\caption{Domain of dependence of an open neighborhood $U_{t_0}$ of the initial hypersurface $\Sigma_{t_0}$.} \label{domBigBang}
\end{figure}

\begin{theorem}\label{Thm:BlowUP_local}
Let $U_{t_0}$ be an open neighborhood in $\Sigma_{t_0}$, $t_0>0$, whose domain of dependence intersects $\Sigma_0$ in $U_0=(0,\delta)^3$, and let 
$\psi$ be a smooth solution to the wave equation, $\square_g\psi=0$, for either of the metrics $g_{\mathrm{FLRW}}, g_{\mathrm{Kasner}}$, arising from initial data $(\psi_0,\partial_t\psi_0)$ on $U_{t_0}$. If $\partial_t\psi_0$ is non-zero in $U_{t_0}$, $t_0$ is sufficiently small such that 
\begin{align}
\label{localsmallt0FLRW}
\frac{2t_0^{2-\frac{4}{3\gamma}}}{1-(\frac{2}{3\gamma})^2}\sum_{i=1}^3\|\partial_t\partial_{x_i}\psi_0\|_{L^2(U_{t_0})}^2
+\frac{4t_0^{1-\frac{2}{3\gamma}}}{1-\frac{2}{3\gamma}}\bigg[
3\|\partial_t\psi_0\|^2_{L^2(U_{t_0})}
+\sum_{l=1}^3\|\partial_t\partial_{x_l}\psi_0\|^2_{L^2(U_{t_0})}\bigg]\notag&\\
+6\log\big(1+\frac{2}{1-\frac{2}{3\gamma}}\frac{t^{1-\frac{2}{3\gamma}}_0}{\delta}\big)\|\partial_t\psi_0\|^2_{L^2(U_{t_0})}<\epsilon\|\partial_t\psi_0\|_{L^2(U_{t_0})},& &\mathrm{(FLRW)}\\
\label{localsmallt0Kasner}\sum_{i=1}^3\frac{2t_0^{2-2p_i}}{(1-p_i)^2}\|\partial_t\partial_{x_i}\psi_0\|_{L^2(U_{t_0})}^2
+\sum_{l=1}^3\frac{4t_0^{1-p_l}}{1-p_l}\big[
\|\partial_t\psi_0\|^2_{L^2(U_{t_0})}
+\|\partial_t\partial_{x_l}\psi_0\|^2_{L^2(U_{t_0})}\big]\notag&\\
+\sum_{l=1}^32\log\big(1+\frac{2}{1-p_l}\frac{t^{1-p_l}_0}{\delta}\big)\|\partial_t\psi_0\|^2_{L^2(U_{t_0})}<\epsilon\| \partial_t\psi_0\|^2_{L^2(U_{t_0})},&& (\mathrm{Kasner})
\end{align}
and $(\psi_0,\partial_t\psi_0)$ satisfy the open conditions:
\begin{align}
\label{localinitcondFLRW}\notag(1-\epsilon)\|\partial_t\psi_0 \|^2_{L^2(U_{t_0})} >&\, 
t^{-\frac{4}{3\gamma}}_0\sum_{i=1}^3\|\partial_{x_i}\psi_0\|^2_{L^2(U_{t_0})}
+\frac{2t_0^{2-\frac{8}{3\gamma}}}{1-(\frac{2}{3\gamma})^2}\sum_{i,j=1}^3\|\partial_{x_j}\partial_{x_i}\psi_0\|_{L^2(U_{t_0})}^2\\
&+\frac{4t_0^{1-\frac{2}{\gamma}}}{1-\frac{2}{3\gamma}}\bigg[
3t_0^{-\frac{4}{3\gamma}}\sum_{i=1}^3\|\partial_{x_i}\psi_0\|^2_{L^2(U_{t_0})}
+t_0^{-\frac{4}{3\gamma}}\sum_{i,l=1}^3\|\partial_{x_i}\partial_{x_l}\psi_0\|^2_{L^2(U_{t_0})}
\bigg]\\
\notag&+6\log\big(1+\frac{2}{1-\frac{2}{3\gamma}}\frac{t^{1-\frac{2}{\gamma}}_0}{\delta}\big)t_0^{-\frac{4}{3\gamma}}\sum_{i=1}^3\|\partial_{x_i}\psi_0\|^2_{L^2(U_{t_0})},&\mathrm{(FLRW)}\\
\label{localinitcondKasner} \notag(1-\epsilon)\|\partial_t\psi_0\|^2_{L^2(U_{t_0})}>&\,\sum_{i=1}^3t_0^{-2p_i}\|\partial_{x_i}\psi_0\|^2_{L^2(U_{t_0})}
+\sum_{i,j=1}^3\frac{2t_0^{2-2p_i-2p_j}}{(1-p_i)^2}\|\partial_{x_j}\partial_{x_i}\psi_0\|_{L^2(U_{t_0})}^2\\
&+\sum_{i,l=1}^3\frac{4t_0^{1-p_l}}{1-p_l}\big[
t_0^{-2p_i}\|\partial_{x_i}\psi_0\|^2_{L^2(U_{t_0})}
+t_0^{-2p_i}\|\partial_{x_i}\partial_{x_l}\psi_0\|^2_{L^2(U_{t_0})}
\big]\\
\notag&+\sum_{l=1}^32\log\big(1+\frac{2}{1-p_l}\frac{t^{1-\frac{2}{\gamma}}_0}{\delta}\big)\sum_{i=1}^3t_0^{-2p_i}\|\partial_{x_i}\psi_0\|^2_{L^2(U_{t_0})},&\mathrm{(Kasner)}
\end{align}
for some $0<\epsilon<1$, then $\|A(x)\|_{L^2(U_0)}>0$.	
\end{theorem}
The blow up behaviour of linear waves observed near Big Bang singularities is reminiscent of the behaviour of waves in black hole interiors containing spacelike singularities \cite{Burko,DorNov}. Examples are the Schwarzschild singularity or black hole singularities occurring in spherically symmetric solutions to the Einstein-scalar field model \cite{christo_sing}, where a logarithmic blow up behaviour has been observed for spatially homogeneous waves. Such logarithmic blow up behaviour was recently confirmed \cite{FourSbier} for generic linear waves in the Schwarzschild black hole interior. The aforementioned blow up behaviours, however, are in contrast to the behaviour of waves observed near null boundaries, where linear and dynamical waves have been shown
in general to extend continuously past the relevant null hypersurfaces \cite{joao_ich}, \cite{anne}-\cite{peter2}, see also~\cite{LP18}.

Lastly, we should note that although we only deal with spatially homogeneous spacetimes, our method of proof is applicable to cosmological spacetimes with Big Bang singularities exhibiting asymptotically velocity term dominated (AVTD) behaviour \cite{ABIF13,Uggla,BIM04,DHRW02,ELS72,HS11,IM02,Rin05}.

\section{Proof of main theorems}

\subsection{Energy argument and notation}

In order to derive the energy estimates required for the proof of Theorem \ref{Thm:BlowUP}, we will apply the vector field method and define certain energy currents constructed from the \emph{stress-energy tensor} 
\begin{align}\label{Tab}
T_{a b}[\psi] = \partial_a \psi \partial_b \psi - \frac{1}{2} g_{a b} \partial^c \psi \partial_c \psi \;,
\end{align}
of the scalar field $\psi$.
The divergence of $T_{ab}[\psi]$ reads
\begin{align}\label{divTab}
\nabla^a T_{a b}[\psi]=\partial_b\psi\cdot\square_g\psi,
\end{align}
where $\nabla$ stands for the spacetime covariant connection. Hence, if $\psi$ satisfies the homogeneous wave equation 
\begin{align}
\square_g\psi=0,
\end{align}
it follows that we have energy-momentum conservation $\nabla^a T_{a b}[\psi] = 0$. Contracting the stress-energy tensor with a vector field multiplier $X$, then defines the associated current 
\begin{align}\label{Ja}
J^X_a[\psi] = X^b T_{a b}[\psi], 
\end{align}
whose divergence, according to (\ref{divTab}), equals
\begin{align}\label{divJa}
\nabla^aJ^X_a[\psi]=(\nabla^aX^b)T_{ab}[\psi]+X\psi\cdot\square_g\psi
\end{align}
Applying the divergence theorem to (\ref{divJa}), over the spacetime domain $\{U_s\}_{s\in[t,t_0]}$ (Figure \ref{domBigBang}), we thus obtain:
\begin{equation}
\label{divthe}
\int_{U_t} J^X_a [\psi]n^a_{U_t} \mathrm{vol}_{U_t}+\sum_{l=1}^3\int_{\cup\mathcal{N}_l^\pm}J^X_a[\psi]n^a_{\mathcal{N}_l^\pm}\mathrm{vol}_{\mathcal{N}_l^\pm}
=\int_{U_{t_0}} J^X_a[\psi]n^a_{U_{t_0}} \mathrm{vol}_{U_{t_0}}-\int_t^{t_0}\int_{U_s} \nabla^a J_a^X[\psi] \mathrm{vol}_{U_s} d s,
\end{equation}
where $n_{U_t}=-\partial_t$, $\mathrm{vol}_{U_t}$ is the intrinsic volume form of $U_t$ and 
\begin{align}\label{n}
\begin{split}
n_{\mathcal{N}_l^\pm}=-\partial_t\pm t^{-\frac{2}{3\gamma}}\partial_{x_l},\qquad\mathrm{vol}_{\mathcal{N}_l^\pm}=t^\frac{4}{3\gamma}dtdx_idx_j,\qquad 
t^{-\frac{2}{3\gamma}}dt=\pm d{x_l},\;\;\text{on $\mathcal{N}_l^\pm$}\qquad\text{(FLRW)}& \\
n_{\mathcal{N}_l^\pm}=-\partial_t\pm t^{-p_l}\partial_{x_l},\qquad\mathrm{vol}_{\mathcal{N}_l^\pm}=t^{1-p_l}dtdx_idx_j,\qquad 
t^{-p_l}dt=\pm d{x_l},\;\;\text{on $\mathcal{N}_l^\pm$}\qquad\text{(Kasner)}&
\end{split}
\end{align}
for each $l=1,2,3$; $i<j$; $i,j\neq l$. 

\sloppy
Below we will choose the vector field $X$ to be a suitable rescaling $n_{U_t}$.  Note that, \mbox{$n^{a}_{U_t} J^X_a[\psi]=J^{X}_{0}[\psi]=\frac{1}{2}[(\partial_t\psi)^2+|\overline{\nabla}\psi|^2]$} and $J^X_a[\psi]n^a_{\mathcal{N}_l^\pm}\ge0$, where $\overline{\nabla}$ is the covariant derivative intrinsic to the level sets of $t$. Hence, if we can control the second term on the RHS of (\ref{divthe}) (the bulk), in terms of $J^X_0[\psi]$, we obtain an energy estimate for $\psi$.

\fussy
Notice that in the case of the whole torus, $U_{t_0}=\Sigma_{t_0}$, due to the absence of causal boundary terms, (\ref{divthe}) becomes 
\begin{equation}
\label{divthe2}
\int_{\Sigma_t} J^X_a [\psi]n^a_{\Sigma_t} \mathrm{vol}_{\Sigma_t}
=\int_{\Sigma_{t_0}} J^X_a[\psi]n^a_{\Sigma_{t_0}} \mathrm{vol}_{\Sigma_{t_0}}-\int_t^{t_0}\int_{\Sigma_s} \nabla^a J_a^X[\psi] \mathrm{vol}_{\Sigma_s} ds.
\end{equation}

In the analysis that follows, we will often require higher order energy estimates that we can obtain by commuting the wave equation with spatial derivatives and applying the above energy argument. Note, that we are considering homogeneous spacetimes in which the spatial coordinate derivatives $\{\partial_{x_i}\}$ are Killing and hence $[\square_g,\partial_{x_i}]=0$, $i=1,2,3$. This means that the above identities (\ref{divthe}), (\ref{divthe2}) are also valid for $\partial_x^\alpha\psi$, where we use the standard multi-index notation $\partial_x^\alpha=\partial_{x_1}^{\alpha_1}\partial_{x_2}^{\alpha_2}\partial_{x_3}^{\alpha_3}$ for an iterated application of spatial derivatives, $\alpha=(\alpha_1,\alpha_2,\alpha_3)$, $|\alpha|=\alpha_1+\alpha_2+\alpha_3$. In this notation, the $H^k(\Sigma_t)$ norm of a smooth function $f:(0,+\infty)\times \mathbb{T}^3$ equals
\begin{align}\label{Hk}
\|f\|_{H^k(\Sigma_t)}^2=\sum_{|\alpha|\leq k}\int_{\Sigma_t}(\partial_x^\alpha f)^2\mathrm{vol}_{Euc},
\end{align}
where $\mathrm{vol}_{Euc}=dx_1dx_2dx_3$. We will often omit $\Sigma_t$ from the norms to ease notation and use $H^k,L^2$ for the corresponding time-dependent, non-intrinsic norms.

\subsection{Flat FLRW}

Let $\psi$ be a smooth solution to the scalar wave equation
\begin{align}\label{boxpsiFLRW}
\square_{g_{\mathrm{FLRW}}}\psi=0.
\end{align}
Consider the orthonormal frame
\begin{align}\label{frame}
e_0=-\partial_t,\qquad e_i=t^{-\frac{2}{3\gamma}}\partial_{x_i}
\end{align}
adapted to the constant $t$ hypersurfaces $\Sigma_t$ with the past normal vector field $e_0$ pointing towards the singularity. In this frame, the second fundamental form $K_{ij}$ of $\Sigma_t$ reads 
\begin{align}\label{KiiFLRW}
K_{ii}:=g(\nabla_{e_i}e_0,e_i)=-\frac{2}{3\gamma}\frac{1}{t},&&i=1,2,3.
\end{align}
Further, the intrinsic volume form on $\Sigma_t$ equals
\begin{align}\label{volFLRW}
\mathrm{vol}_{\Sigma_t}=t^{\frac{2}{\gamma}}\mathrm{vol}_{Euc}.
\end{align}
\begin{proposition}\label{prop:upperestFLRW}
The following energy inequality holds: 
\begin{align}\label{psienestFLRW}
t^{\frac{2}{\gamma}}\int_{\Sigma_t}J^{e_0}_0[\partial_x^\alpha\psi]\mathrm{vol}_{\Sigma_t}\leq t^\frac{2}{\gamma}_0\int_{\Sigma_{t_0}}J^{e_0}_0[\partial_x^\alpha\psi]\mathrm{vol}_{\Sigma_{t_0}},
\end{align}
for all $t\in(0,t_0]$ and any multi-index $\alpha$. Moreover, $\psi$ satisfies the pointwise bound
\begin{align}\label{LinftypsiFLRW}
|\psi(t,x)|\leq C\bigg(\sum_{|\alpha|\leq2}t^\frac{2}{\gamma}_0\int_{\Sigma_{t_0}}J^{e_0}_0[\partial_x^\alpha\psi]\mathrm{vol}_{\Sigma_{t_0}}\bigg)^\frac{1}{2} \frac{t^{1-\frac{2}{\gamma}}-t_0^{1-\frac{2}{\gamma}}}{\frac{2}{\gamma}-1}+|\psi(t_0,x)|,
\end{align}
where $C>0$ is a constant independent of $t_0,\gamma$.
\end{proposition}
\begin{proof}
We compute the divergence of $J^{t^\frac{2}{\gamma}e_0}_a[\psi]$:
\begin{align}\label{divJt2/gammae0}
\nabla^aJ^{t^\frac{2}{\gamma}e_0}_a[\psi]\overset{\eqref{divJa}}{=}&\,\nabla^a(t^\frac{2}{\gamma}e_0)^bT_{ab}[\psi]=t^\frac{2}{\gamma}K^{ab}T_{ab}[\psi]-e_0t^\frac{2}{\gamma}T_{00}[\psi]\\
\notag=&\,t^\frac{2}{\gamma}K_{11}|\overline{\nabla}\psi|^2-\frac{1}{2}t^\frac{2}{\gamma}{K^i}_i|\nabla\psi|^2+\frac{1}{\gamma}t^{\frac{2}{\gamma}-1}[(e_0\psi)^2+|\overline{\nabla}\psi|^2]\\
\notag\overset{\eqref{KiiFLRW}}{=}&\,\frac{4}{3\gamma}\frac{1}{t}t^\frac{2}{\gamma}|\overline{\nabla}\psi|^2.
\end{align}
Hence, utilising (\ref{divthe2}) for $X=t^\frac{2}{\gamma}e_0$ yields
\begin{align}\label{Stokes1FLRW}
t^\frac{2}{\gamma} \int_{\Sigma_t}J^{e_0}_0[\psi]\mathrm{vol}_{\Sigma_t}=&\,t^\frac{2}{\gamma}_0\int_{\Sigma_{t_0}}J^{e_0}_0[\psi]\mathrm{vol}_{\Sigma_{t_0}}-\int^{t_0}_t\int_{\Sigma_s}\frac{4}{3\gamma}s^{\frac{2}{\gamma}-1}|\overline{\nabla}\psi|^2\mathrm{vol}_{\Sigma_s} ds. \\
\leq &\, t^\frac{2}{\gamma}_0\int_{\Sigma_{t_0}}J^{e_0}_0[\psi]\mathrm{vol}_{\Sigma_{t_0}}
\end{align}
The same identity is valid for $\partial_x^\alpha\psi$, leading to (\ref{psienestFLRW}). 
In particular, taking into account the volume form (\ref{volFLRW}), we have the following bounds for $\partial_t\psi$:
\begin{align}\label{e0psiest}
t^{\frac{4}{\gamma}}\|\partial_t\partial_x^\alpha\psi\|^2_{L^2}\leq 2t^\frac{2}{\gamma}_0\int_{\Sigma_{t_0}}J^{e_0}_0[\partial^\alpha_x\psi]\mathrm{vol}_{\Sigma_{t_0}},
\end{align}
for all $t\in(0,t_0]$ and $\alpha$. Integrating $\partial_t\psi$ in $[t,t_0]$ and employing the Sobolev embedding $H^2(\mathbb{T}^3)\hookrightarrow L^\infty(\mathbb{T}^3)$ we derive:
\begin{align}\label{Linftypsi}
|\psi(t,x)|=&\,\big|\int^t_{t_0}\partial_s\psi(s,x) ds+\psi(t_0,x)\big|\\
\notag\leq&\,C\int^{t_0}_t\|\partial_s\psi\|_{H^2} ds+|\psi(t_0,x)|\\
\tag{$\gamma<2$}\leq&\,\frac{C}{\frac{2}{\gamma}-1}(t^{1-\frac{2}{\gamma}}-t_0^{1-\frac{2}{\gamma}}) (\sum_{|\alpha|\leq2}t^\frac{2}{\gamma}_0\int_{\Sigma_{t_0}}J^{e_0}_0[\partial_x^\alpha\psi]\mathrm{vol}_{\Sigma_{t_0}})^{\frac{1}{2}}+|\psi(t_0,x)|.
\end{align}
\end{proof}
\begin{remark}
The bounds (\ref{psienestFLRW}), (\ref{LinftypsiFLRW}) are saturated by the homogeneous function $t^{1-\frac{2}{\gamma}}$, which is an exact solution of (\ref{boxpsiFLRW}).
\end{remark}
From the previous proposition we understand that $t^{1-\frac{2}{\gamma}}$ is the leading order of $\psi$ at $t=0$. To prove this rigorously, we derive analogous energy bounds for the renormalised variable $\frac{\psi}{t^{1-\frac{2}{\gamma}}}$ that satisfies the wave equation:
\begin{align}\label{boxvarphi/t^1-2/gamma}
\square\frac{\psi}{t^{1-\frac{2}{\gamma}}}=-\frac{2}{t}(1-\frac{2}{\gamma})e_0(\frac{\psi}{t^{1-\frac{2}{\gamma}}}).
\end{align}
\begin{proposition}\label{prop:renestFLRW}
Let $\psi$ be a smooth solution to the wave equation in FLRW backgrounds with $\frac23<\gamma<2$. Then, the following bounds hold uniformly in $t\in(0,t_0]$:
\begin{align}
\label{renenestFLRW}t^{4-\frac{6}{\gamma}}\int_{\Sigma_t}J^{e_0}_0 [\partial_x^\alpha\frac{\psi}{t^{1-\frac{2}{\gamma}}} ]\mathrm{vol}_{\Sigma_t}\leq t^{4-\frac{6}{\gamma}}_0\int_{\Sigma_{t_0}}J^{e_0}_0 [\partial_x^\alpha\frac{\psi}{t^{1-\frac{2}{\gamma}}} ]\mathrm{vol}_{\Sigma_{t_0}},\qquad\frac{4}{3}\leq\gamma<2,\\
\label{renenestFLRW_softer} 
t^{-\frac{2}{3\gamma}}\int_{\Sigma_t}J^{e_0}_0 [\partial_x^\alpha\frac{\psi}{t^{1-\frac{2}{\gamma}}} ]\mathrm{vol}_{\Sigma_t}\leq t^{-\frac{2}{3\gamma}}_0\int_{\Sigma_{t_0}}J^{e_0}_0 [\partial_x^\alpha\frac{\psi}{t^{1-\frac{2}{\gamma}}} ]\mathrm{vol}_{\Sigma_{t_0}},\qquad\frac{2}{3}<\gamma\leq\frac{4}{3},
\end{align}
for all $t\in(0,t_0]$ and any multi-index $\alpha$. 
Moreover, the limit
\begin{align}\label{AdefFLRW}
A(x):=\lim_{t\rightarrow0}\frac{\psi}{t^{1-\frac{2}{\gamma}}}
\end{align}
exists, it is a smooth function and the difference $u(t,x):=\psi-A(x){t^{1-\frac{2}{\gamma}}}$ satisfies
\begin{align}\label{diffenestFLRW}
\lim_{t\rightarrow0}t^\frac{2}{\gamma}\int_{\Sigma_t}J^{e_0}_0[\partial^\alpha_xu]\mathrm{vol}_{\Sigma_s}=0.
\end{align}
\end{proposition}

\begin{proof}
Let $\eta>0$. We compute
\begin{align}\label{divJren}
\nabla^a(J^{t^\eta e_0}_a[\frac{\psi}{t^{1-\frac{2}{\gamma}}}])\overset{\eqref{divJa}}{=}&\,\nabla^a(t^\eta e_0)^bT_{ab}[\frac{\psi}{t^{1-\frac{2}{\gamma}}}]+t^\eta e_0(\frac{\psi}{t^{1-\frac{2}{\gamma}}})\cdot\square_g\frac{\psi}{t^{1-\frac{2}{\gamma}}}\\
\notag=&\;t^\eta K^{ab}T_{ab}[\frac{\psi}{t^{1-\frac{2}{\gamma}}}]-(e_0t^\eta )T_{00}[\frac{\psi}{t^{1-\frac{2}{\gamma}}}]-\frac{2}{t}(1-\frac{2}{\gamma})t^\eta [e_0(\frac{\psi}{t^{1-\frac{2}{\gamma}}})]^2\\
\notag\overset{\eqref{KiiFLRW}}{=}&\,t^{\eta-1}\bigg[(\frac{1}{3\gamma}+\frac{\eta}{2})|\overline{\nabla}\frac{\psi}{t^{1-\frac{2}{\gamma}}}|^2+(\frac{\eta}{2}+\frac{3}{\gamma}-2) [e_0(\frac{\psi}{t^{1-\frac{2}{\gamma}}})]^2
\bigg]
\end{align}
This leads to different choices of $\eta$ depending on the value of $\gamma$, given by
\bea
\lb{stiff}
 \eta &=&4-\frac{6}{\gamma},  \quad \mbox{for} \quad\frac{4}{3}\leq\gamma<2, \\
\lb{soft}
 \eta &=&-\frac{2}{3\gamma}, \quad \mbox{for} \quad \frac{2}{3}<\gamma\leq\frac{4}{3}. 
\eea
We refer to the former as the stiffest region and the latter as the softest region. The case $\gamma=\frac{4}{3}$ corresponds to radiation, where $\eta=-\frac{1}{2}$.
For the two cases, (\ref{divJren}) reads
\begin{align}
\label{divJrenstiff}\nabla^a(J^{t^{4-\frac{6}{\gamma}} e_0}_a[\frac{\psi}{t^{1-\frac{2}{\gamma}}}])=&\,(2-\frac{8}{3\gamma})t^{3-\frac{6}{\gamma}}|\overline{\nabla}\frac{\psi}{t^{1-\frac{2}{\gamma}}}|^2,\\
\label{divJrensoft}\nabla^a(J^{t^{-\frac{2}{3\gamma}} e_0}_a[\frac{\psi}{t^{1-\frac{2}{\gamma}}}])=&\,(\frac{8}{3\gamma}-2)t^{-1-\frac{2}{3\gamma}}[e_0(\frac{\psi}{t^{1-\frac{2}{\gamma}}})]^2.
\end{align}
Then, the energy identity (\ref{divthe2}) for $\frac{\psi}{t^{1-\frac{2}{\gamma}}}$,  $X=t^{4-\frac{6}{\gamma}}e_0$, in the stiffest region $\{\frac{4}{3}\leq\gamma<2 \}$, implies that
\begin{align}\label{StokesJtildeFL}
t^{4-\frac{6}{\gamma}}\int_{\Sigma_t}J^{e_0}_0 [\frac{\psi}{t^{1-\frac{2}{\gamma}}} ]\mathrm{vol}_{\Sigma_t}
&\leq
t^{4-\frac{6}{\gamma}}_0\int_{\Sigma_{t_0}}J^{e_0}_0 [\frac{\psi}{t^{1-\frac{2}{\gamma}}} ]\mathrm{vol}_{\Sigma_{t_0}} 
- (2-\frac{8}{3\gamma} )
\int^{t_0}_t s^{3-\frac{6}{\gamma}}\int_{\Sigma_s} |\overline{\nabla}\frac{\psi}{t^{1-\frac{2}{\gamma}}}|^2\mathrm{vol}_{\Sigma_s} ds \nonumber \\
&\leq t^{4-\frac{6}{\gamma}}_0\int_{\Sigma_{t_0}}J^{e_0}_0 [\frac{\psi}{t^{1-\frac{2}{\gamma}}_0} ]\mathrm{vol}_{\Sigma_{t_0}}. 
\end{align}
Commuting with $\partial_{x}^\alpha$ yields the bound \eqref{renenestFLRW}.
Similarly, we obtain statement \eqref{renenestFLRW_softer} in the softest region $ \{\frac{2}{3}<\gamma\leq\frac{4}{3} \}$. In particular, taking into account the volume form (\ref{volFLRW}), we have the bounds:
\begin{align}
 |\partial_t\frac{\psi}{t^{1-\frac{2}{\gamma}}} |\leq C\|\partial_t\frac{\psi}{t^{1-\frac{2}{\gamma}}}\|_{H^2}\leq \frac{C}{t^{2-\frac{2}{\gamma}}}\bigg(\sum_{|\alpha|\leq 2}t^{4-\frac{6}{\gamma}}_0\int_{\Sigma_{t_0}}J^{e_0}_0 [\partial_x^\alpha\frac{\psi}{t^{1-\frac{2}{\gamma}}_0} ]\mathrm{vol}_{\Sigma_{t_0}}\bigg)^\frac{1}{2},&&\frac{4}{3}\leq\gamma<2 \\
 |\partial_t\frac{\psi}{t^{1-\frac{2}{\gamma}}} |\leq C\|\partial_t\frac{\psi}{t^{1-\frac{2}{\gamma}}}\|_{H^2}\leq \frac{C}{t^{\frac{2}{3\gamma}}}\bigg(\sum_{|\alpha|\leq 2}t^{-\frac{2}{3\gamma}}_0\int_{\Sigma_{t_0}}J^{e_0}_0 [\partial_x^\alpha\frac{\psi}{t^{1-\frac{2}{\gamma}}} ]\mathrm{vol}_{\Sigma_{t_0}}\bigg)^\frac{1}{2},&&\frac{2}{3}<\gamma\leq\frac{4}{3}
\end{align}
which imply that $\partial_t\psi(t,x) \in L^1([0,t_0])$, uniformly in $x$, for all $\frac{2}{3}<\gamma<2$. Thus, $\frac{\psi}{t^{1-\frac{2}{\gamma}}}$ has a limit function $A(x)$, as $t\rightarrow0$. The smoothness of $A(x)$ follows by repeating the preceding argument for $\partial_x^{\alpha}\frac{\psi}{t^{1-\frac{2}{\gamma}}}$. 

Consider now the energy flux of the difference $\psi-A(x) t^{1-\frac{2}{\gamma}}$,
\begin{align}\label{diffenest2FL}
&t^{\frac{2}{\gamma}}\int_{\Sigma_t} |e_0 (\psi-A(x)t^{1-\frac{2}{\gamma}} ) |^2+ |\overline{\nabla} (\psi-A(x)t^{1-\frac{2}{\gamma}} ) |^2\mathrm{vol}_{\Sigma_t}\\
\notag=&\,t^{\frac{2}{\gamma}}\int_{\Sigma_t} |t^{1-\frac{2}{\gamma}}\,e_0\frac{\psi}{t^{1-\frac{2}{\gamma}}}+ (1-\frac{2}{\gamma} )t^{-\frac{2}{\gamma}} (\frac{\psi}{t^{1-\frac{2}{\gamma}}}-A(x) ) |^2+t^{2-\frac{4}{\gamma}} |\overline{\nabla} (\frac{\psi}{t^{1-\frac{2}{\gamma}}}-A(x) ) |^2\mathrm{vol}_{\Sigma_t}\\
\notag\leq&\,2t^{\frac{2}{\gamma}}\int_{\Sigma_t}t^{2-\frac{4}{\gamma}} |e_0\frac{\psi}{t^{1-\frac{2}{\gamma}}} |^2+ (1-\frac{2}{\gamma} )^2 t^{-\frac{4}{\gamma}} |\frac{\psi}{t^{1-\frac{2}{\gamma}}}-A(x) |^2+t^{2-\frac{4}{\gamma}} |\overline{\nabla}\frac{\psi}{t^{1-\frac{2}{\gamma}}} |^2+t^{2-\frac{4}{\gamma}} |\overline{\nabla}A(x) |^2\mathrm{vol}_{\Sigma_t}\\
\notag=& \,4t^2\int_{\Sigma_t} J^{e_0}_0 [\frac{\psi}{t^{1-\frac{2}{\gamma}}} ]\mathrm{vol}_{Euc}+2 (1-\frac{2}{\gamma} )^2\int_{\Sigma_t} |\frac{\psi}{t^{1-\frac{2}{\gamma}}}-A(x) |^2\mathrm{vol}_{Euc}+2t^2\int_{\Sigma_t}|\overline{\nabla}A(x)|^2 \mathrm{vol}_{Euc} \\
\notag&\overset{\eqref{renenestFLRW},\eqref{renenestFLRW_softer} }\leq o(1)+2 (1-\frac{2}{\gamma} )^2\int_{\Sigma_t} |\frac{\psi}{t^{1-\frac{2}{\gamma}}}-A(x) |^2\mathrm{vol}_{Euc}+2t^{2-\frac{4}{3\gamma}}\sum_{i=1}^3\int_{\Sigma_t}|\partial_{x_i}A(x)|^2 \mathrm{vol}_{Euc},
\end{align}
for all $\frac{4}{3}\leq\gamma<2$. The third term in the preceding RHS clearly tends to zero, as $t\rightarrow0$, and by the definition of $A(x)$, so does the second term. Since the above argument also applies to $\partial^\alpha_x[\psi-A(x)t^{1-\frac{2}{\gamma}}]$, this proves \eqref{diffenestFLRW}.
\end{proof}
\begin{remark}
The renormalised estimate (\ref{renenestFLRW}) yields an improved control over the spatial gradient of $\psi$ compared to (\ref{psienestFLRW}). Indeed,  
\begin{align}\label{gradbarpsiestFLRW}
\begin{split}
t^{2-\frac{4}{\gamma}}t^\frac{2}{\gamma}\int_{\Sigma_t}|\overline{\nabla}\psi|^2\mathrm{vol}_{\Sigma_t}\leq t^{4-\frac{6}{\gamma}}_0\int_{\Sigma_{t_0}}J^{e_0}_0 [\frac{\psi}{t^{1-\frac{2}{\gamma}}} ]\mathrm{vol}_{\Sigma_{t_0}},\qquad\frac{4}{3}\leq\gamma<2,\\
 t^{\frac{4}{3\gamma}-2}t^\frac{2}{\gamma}\int_{\Sigma_t}|\overline{\nabla}\psi|^2\mathrm{vol}_{\Sigma_t}\leq t^{-\frac{2}{3\gamma}}_0\int_{\Sigma_{t_0}}J^{e_0}_0 [\frac{\psi}{t^{1-\frac{2}{\gamma}}} ]\mathrm{vol}_{\Sigma_{t_0}},\qquad\frac{2}{3}<\gamma\leq\frac{4}{3},
\end{split}
\end{align}
holds for all $t\in(0,t_0]$, where in the stiffest case $2-\frac{4}{\gamma}<0$, while in the softest $\frac{4}{3\gamma}-2<0$.
\end{remark}
The previous proposition validates the asymptotic profile (\ref{psiasymFLRW}) of $\psi$, as stated in Theorem \ref{Thm:Asym}.
\begin{lemma}\label{lem:L2estFLRW}
The following estimate for the $L^2$ norm of $\partial_{x_i}\psi$ holds:
\begin{align}\label{L2estFLRW}
 \|\partial_{x_i}\psi\|_{L^2(\Sigma_t)}\leq \|\partial_{x_i}\psi\|_{L^2(\Sigma_{t_0})}+\sqrt{2}\frac{ (t^{1-\frac{2}{\gamma}}_0-t^{1-\frac{2}{\gamma}} )}{ 1-\frac{2}{\gamma} } t^{\frac{2}{\gamma}}_0 \bigg(\int_{\Sigma_{t_0}}J^{e_0}_0[\partial_{x_i}\psi]\mathrm{vol}_{Euc} \bigg)^{\frac{1}{2}},
\end{align}
for all $t\in(0,t_0]$.
\end{lemma}
\begin{proof}
Differentiating in $e_0$ we have:
\begin{align}\label{L2est}
 \frac{1}{2}e_0\|\partial_{x_i}\psi\|^2_{L^2(\Sigma_t)}\leq&\, \|\partial_{x_i}\psi\|_{L^2(\Sigma_t)}\|e_0\psi\|_{L^2(\Sigma_t)}\\
\tag{using (\ref{psienestFLRW})}\leq&\, \frac{1}{t^{\frac{2}{\gamma}}}\bigg(2t^{\frac{2}{\gamma}}_0\int_{\Sigma_{t_0}}J^{e_0}_0[\partial_{x_i}\psi]\mathrm{vol}_{\Sigma_{t_0}} \bigg)^{\frac{1}{2}}\|\partial_{x_i}\psi\|_{L^2(\Sigma_t)}
\end{align}
or
\begin{align}\label{L2est2}
e_0 \|\partial_{x_i}\psi\|_{L^2(\Sigma_t)}\leq \sqrt{2} \frac{t^{\frac{2}{\gamma}}_0}{t^{\frac{2}{\gamma}}}\bigg(\int_{\Sigma_{t_0}}J^{e_0}_0[\partial_{x_i}\psi]\mathrm{vol}_{Euc} \bigg)^{\frac{1}{2}}.
\end{align}
Integrating the above on $[t,t_0]$ gives (\ref{L2estFLRW}) for $\gamma<2$.
\end{proof}
\begin{remark}
The bounds that we have proven so far, stated in Propositions \ref{prop:upperestFLRW}, \ref{prop:renestFLRW} and Lemma \ref{lem:L2estFLRW}, are also valid if we replace the integral domains $\Sigma_t,\Sigma_{t_0}$ by $U_t,U_{t_0}$. This can be easily seen from the fact that in the corresponding energy identity (\ref{divthe}), the null boundary terms have a favourable sign for an upper bound and therefore can be dropped.
\end{remark}
Now we may proceed to derive the blow up criterion given in Theorem \ref{Thm:BlowUP}. First, notice that 
for $\gamma>\frac{2}{3}$, the main contribution of the energy flux generated by $J^{e_0}[\psi]$ comes from the $e_0\psi$ term. Indeed, by (\ref{gradbarpsiestFLRW}) it follows that
\begin{align}\label{varphienbehFL}
t^{\frac{2}{\gamma}}\int_{\Sigma_t}J^{e_0}_0[\psi]\mathrm{vol}_{\Sigma_t}
=&\,\frac{1}{2}\int_{\Sigma_t}t^{\frac{4}{\gamma}}(\partial_t\psi)^2\mathrm{vol}_{Euc}+O(t^\eta)
\end{align}
where $\eta=\frac{4}{\gamma}-2>0$, for $\gamma\in[\frac{4}{3},2)$ and $\eta=2-\frac{4}{3\gamma}>0$, for $\gamma\in(\frac{2}{3},\frac{4}{3}]$.
Hence, taking the limit $t\rightarrow0$ in the preceding identity and utilizing \eqref{diffenestFLRW} leads to
\begin{align}\label{Jlim}
\lim_{t\rightarrow0}t^{\frac{2}{\gamma}}\int_{\Sigma_t}J^{e_0}_0[\psi]\mathrm{vol}_{\Sigma_t}
=\frac{1}{2} (1-\frac{2}{\gamma} )^2\int_{\Sigma_0}A^2(x)\mathrm{vol}_{Euc}.
\end{align}
Combining \eqref{Stokes1FLRW}, \eqref{Jlim} we derive:
\begin{align}\label{lowboundFLRW}
&\frac{1}{2} (1-\frac{2}{\gamma} )^2\int_{\Sigma_0}A^2(x)\mathrm{vol}_{Euc} \\
\notag=&\,t^{\frac{2}{\gamma}}_0\int_{\Sigma_{t_0}}J^{e_0}_0[\psi]\mathrm{vol}_{\Sigma_{t_0}}-\frac{4}{3\gamma}\int^{t_0}_0s^{\frac{2}{\gamma}-1}\int_{\Sigma_s}|\overline{\nabla}\psi|^2 \mathrm{vol}_{\Sigma_s} ds \\
\tag{by \eqref{L2estFLRW}  }\ge&\,\frac{1}{2}t^{\frac{4}{\gamma}}_0\|\partial_t\psi\|_{L^2(\Sigma_{t_0})}^2+\frac{1}{2}t^{\frac{4}{\gamma}-\frac{4}{3\gamma}}_0\sum_{i=1}^3\|\partial_{x_i}\psi\|^2_{L^2(\Sigma_{t_0})}
- \frac{8}{3\gamma}\int^{t_0}_0 s^{\frac{4}{\gamma}-1-\frac{4}{3\gamma}} ds \sum^3_{i=1}\|\partial_{x_i}\psi\|^2_{L^2(\Sigma_{t_0})} \\
\notag&-\frac{16}{3\gamma}\int^{t_0}_0 s^{\frac{4}{\gamma}-1-\frac{4}{3\gamma}} \frac{ (t^{1-\frac{2}{\gamma}}_0-s^{1-\frac{2}{\gamma}} )^2}{ (1-\frac{2}{\gamma} )^2} ds \sum^{3}_{i=1}t^{\frac{2}{\gamma}}_0\int_{\Sigma_{t_0}}J^{e_0}_0[\partial_{x_i}\psi]\mathrm{vol}_{\Sigma_{t_0}}  \\
\notag =&\, \frac{1}{2}t^{\frac{4}{\gamma}}_0\|\partial_t\psi \|^2_{L^2(\Sigma_{t_0})}-\frac{1}{2}t^{\frac{8}{3\gamma}}_0\sum^3_{i=3}\|\partial_{x_i}\psi\|^2_{L^2(\Sigma_{t_0})}  \\
\notag&-\frac{t_0^{2-\frac{4}{3\gamma}}}{1-(\frac{2}{3\gamma})^2}\sum_{i=1}^3\bigg[t_0^\frac{4}{\gamma}\|\partial_t\partial_{x_i}\psi\|^2_{L^2(\Sigma_{t_0})}+t_0^\frac{8}{3\gamma}\sum_{j=1}^3\|\partial_{x_j}\partial_{x_i}\psi\|^2_{L^2(\Sigma_{t_0})}\bigg].
\end{align}
Now is evident now that if the assumptions of Theorem \ref{Thm:BlowUP} for FLRW are satisfied, then $\|A(x)\|_{L^2(\mathbb{T}^3)}>0$.

To prove Theorem \ref{Thm:BlowUP_local} for FLRW, we use the local energy identity (\ref{divthe}) and plug in (\ref{n}), (\ref{volFLRW}), (\ref{divJt2/gammae0}):
\begin{align}\label{Stokes2FLRW}
t^{\frac{2}{\gamma}}\int_{U_t}J^{e_0}_0[\psi]\mathrm{vol}_{U_t}= &\,\frac{1}{2}t^\frac{4}{\gamma}_0\|\partial_t\psi\|^2_{L^2(U_{t_0})}+\frac{1}{2}t^{\frac{4}{\gamma}-\frac{4}{3\gamma}}_0\sum_{i=1}^3\|\partial_{x_i}\psi\|^2_{L^2(U_{t_0})}\\
\notag&-\frac{4}{3\gamma}\int^{t_0}_t s^{\frac{4}{\gamma}-1-\frac{4}{3\gamma}}\sum_{i=1}^3\|\partial_{x_i}\psi\|^2_{L^2(U_s)}ds\\
\tag{$i<j;\,i,j\neq l$}&-\sum_{l=1}^3\int^{t_0}_t\int_{\mathcal{N}_l^\pm\cap \overline{U}_s} s^{\frac{2}{\gamma}+\frac{4}{3\gamma}}\frac{1}{2}\big[|(e_0\pm e_l)\psi|^2+|e_i\psi|^2+|e_j\psi|^2\big]dsdx_idx_j.
\end{align}
Since $t^{-\frac{2}{3\gamma}}dt=\pm dx_l$ along $\mathcal{N}_l^\pm$, it follows by integrating that the closure $\overline{U}_t$ of the neighbourhood $U_t$ is the cube $I^3$, where $I=[-\frac{t^{1-\frac{2}{3\gamma}}}{1-\frac{2}{3\gamma}},\delta+\frac{t^{1-\frac{2}{3\gamma}}}{1-\frac{2}{3\gamma}}]$. We make use of the following 1-dimensional Sobolev inequality\footnote{Proof by fundamental theorem of calculus: $f^2(t,x_l)\leq \min_{x_l\in I}f^2(t,x_l)+\int_I2|f||\partial_{x_l}f|dx_l\leq |I|^{-1}\|f\|_{L^2(I)}^2+\|f\|_{H^1(I)}$.}
\begin{equation}\label{Sob1D}
f^2(t,x_l)\leq (\delta +\frac{2}{1-\frac{2}{3\gamma}}t^{1-\frac{2}{3\gamma}})^{-1}\int_{I}f^2(t,x_l) dx_l+ \|f(t,x_l) \|^2_{H^1(I)},\qquad f\in H^1(I),\;\;t\in(0,t_0].
\end{equation}
First, we take the limit $t\rightarrow0$ in (\ref{Stokes2FLRW}), employing (\ref{Jlim}),
and then we apply (\ref{L2estFLRW}) to the integrand in the third line of (\ref{Stokes2FLRW}) and (\ref{Sob1D}) to the integral over $\mathcal{N}_l^\pm\cap \overline{U}_s$ in the last line of (\ref{Stokes2FLRW}) to deduce the lower bound:
\begin{align}\label{Stokes3FLRW}
&\frac{1}{2}(1-\frac{2}{\gamma})^2\int_{U_0}A^2(x)\mathrm{vol}_{U_0}\\
\notag\ge &\,\frac{1}{2}t^\frac{4}{\gamma}_0\|\partial_t\psi\|^2_{L^2(U_{t_0})}+\frac{1}{2}t^{\frac{4}{\gamma}-\frac{4}{3\gamma}}_0\sum_{i=1}^3\|\partial_{x_i}\psi\|^2_{L^2(U_{t_0})}
- \frac{8}{3\gamma}\int^{t_0}_0 s^{\frac{4}{\gamma}-1-\frac{4}{3\gamma}} ds \sum^3_{i=1}\|\partial_{x_i}\psi\|^2_{L^2(U_{t_0})} \\
\notag&-\frac{16}{3\gamma}\int^{t_0}_0 s^{\frac{4}{\gamma}-1-\frac{4}{3\gamma}} \frac{ (t^{1-\frac{2}{\gamma}}_0-s^{1-\frac{2}{\gamma}} )^2}{ (1-\frac{2}{\gamma} )^2} ds \sum^{3}_{i=1}t^{\frac{2}{\gamma}}_0\int_{U_{t_0}}J^{e_0}_0[\partial_{x_i}\psi]\mathrm{vol}_{U_{t_0}} \\
\notag&-\int^{t_0}_0(1+\frac{ 1}{ \delta+\frac{2}{1-\frac{2}{3\gamma}}s^{1-\frac{2}{3\gamma}} })s^{-\frac{2}{3\gamma}}\int_{U_s}s^\frac{4}{\gamma}12J^{e_0}_0[\psi] \mathrm{vol}_{Euc}ds\\
\notag&-\int^{t_0}_0s^{-\frac{2}{3\gamma}}\int_{U_s} 4s^\frac{4}{\gamma}\sum_{l=1}^3J^{e_0}_0[\partial_{x_l}\psi] \mathrm{vol}_{Euc}\\
\notag\ge&\,\frac{1}{2}t^\frac{4}{\gamma}_0\|\partial_t\psi\|^2_{L^2(U_{t_0})}-\frac{1}{2}t^{\frac{8}{3\gamma}}_0\sum_{i=1}^3\|\partial_{x_i}\psi\|^2_{L^2(U_{t_0})}\\
\notag&-\frac{t_0^{2-\frac{4}{3\gamma}}}{1-(\frac{2}{3\gamma})^2}\sum_{i=1}^3\bigg[t_0^\frac{4}{\gamma}\|\partial_t\partial_{x_i}\psi\|^2_{L^2(\Sigma_{t_0})}+t_0^\frac{8}{3\gamma}\sum_{j=1}^3\|\partial_{x_j}\partial_{x_i}\psi\|_{L^2(\Sigma_{t_0})}^2\bigg]\\
\tag{by (\ref{psienestFLRW}) for $\{U_t\}$}&-\int^{t_0}_0(1+\frac{ 1}{ \delta+\frac{2}{1-\frac{2}{3\gamma}}s^{1-\frac{2}{3\gamma}} })s^{-\frac{2}{3\gamma}}ds
\int_{U_{t_0}}12t_0^\frac{4}{\gamma}J^{e_0}_0[\psi] \mathrm{vol}_{Euc}\\
\notag&-\int^{t_0}_0s^{-\frac{2}{3\gamma}}ds\int_{U_{t_0}} 4t_0^\frac{4}{\gamma}\sum_{l=1}^3J^{e_0}_0[\partial_{x_l}\psi] \mathrm{vol}_{Euc}\\
\notag=&\,\frac{1}{2}t^\frac{4}{\gamma}_0\|\partial_t\psi\|^2_{L^2(U_{t_0})}-\frac{1}{2}t^{\frac{8}{3\gamma}}_0\sum_{i=1}^3\|\partial_{x_i}\psi\|^2_{L^2(U_{t_0})}\\
\notag&-\frac{t_0^{2-\frac{4}{3\gamma}}}{1-(\frac{2}{3\gamma})^2}\sum_{i=1}^3\bigg[t_0^\frac{4}{\gamma}\|\partial_t\partial_{x_i}\psi\|^2_{L^2(\Sigma_{t_0})}+t_0^\frac{8}{3\gamma}\sum_{j=1}^3\|\partial_{x_j}\partial_{x_i}\psi\|_{L^2(\Sigma_{t_0})}^2\bigg]\\
\notag&-\frac{t_0^{1-\frac{2}{3\gamma}}}{1-\frac{2}{3\gamma}}\bigg[
6t^\frac{4}{\gamma}_0\|\partial_t\psi\|^2_{L^2(U_{t_0})}+6t^{\frac{8}{3\gamma}}_0\sum_{i=1}^3\|\partial_{x_i}\psi\|^2_{L^2(U_{t_0})}\\
\notag&+2t^\frac{4}{\gamma}_0\sum_{l=1}^3\|\partial_t\partial_{x_l}\psi\|^2_{L^2(U_{t_0})}+2t^{\frac{8}{3\gamma}}_0\sum_{i,l=1}^3\|\partial_{x_i}\partial_{x_l}\psi\|^2_{L^2(U_{t_0})}
\bigg]\\
\notag&-3\log\big(1+\frac{2}{1-\frac{2}{3\gamma}}\frac{t^{1-\frac{2}{3\gamma}}_0}{\delta}\big)\big[
t^\frac{4}{\gamma}_0\|\partial_t\psi\|^2_{L^2(U_{t_0})}+t^{\frac{8}{3\gamma}}_0\sum_{i=1}^3\|\partial_{x_i}\psi\|^2_{L^2(U_{t_0})}
\big].
\end{align}
Thus, if the assumptions of Theorem \ref{Thm:BlowUP_local} for FLRW are satisfied, then the RHS of (\ref{Stokes3FLRW}) gives $\|A(x)\|_{L^2(U_0)}>0$. This completes the proofs of the main theorems for FLRW.

\subsection{Kasner}

For Kasner the adapted orthonormal frame to the constant $t$ hypersurfaces reads:
\begin{align}\label{frame}
e_0=-\partial_t,\qquad e_i=t^{-p_i}\partial_{x_i}.
\end{align}
In this frame, the non-zero components of the second fundamental form $K$ of the $\Sigma_t$ hypersurfaces are
\begin{align}\label{KasnerKii}
K_{ii}:=g(\nabla_{e_i}e_0,e_i)=-\frac{p_i}{t},&&i=1,2,3.
\end{align}
Further, the intrinsic volume form $\mathrm{vol}_{\Sigma_t}$ on $\Sigma_t$ equals
\begin{align}\label{Kasnerdmu}
\mathrm{vol}_{\Sigma_t}=t \mathrm{vol}_{Euc}.
\end{align}
\begin{proposition}[Upper bound]\label{prop:upperestKasner}
Let $\psi$ be a smooth solution to the wave equation, $\square_g\psi=0$, in Kasner. Then the following energy inequality holds:
\begin{align}\label{psienestKasner}
t\int_{\Sigma_t}J^{e_0}_0[\partial^{\alpha}_x\psi]\mathrm{vol}_{\Sigma_t} \leq t_0\int_{\Sigma_{t_0}}J^{e_0}_0[\partial^{\alpha}_x\psi]\mathrm{vol}_{\Sigma_{t_0}}
\end{align}
for all $t\in(0,t_0]$ and any multi-index $\alpha$. Moreover, $\partial^\beta_x\psi$ satisfies the pointwise bound
\begin{align}\label{LinftypsiKasner}
|\partial^\beta_x\psi(t,x)|\leq C\sum_{|\alpha|\leq 2}\bigg(t_0^2\int_{\Sigma_{t_0}}J^{e_0}_0[\partial_{x}^{\alpha}\partial^\beta_x\psi]\mathrm{vol}_{\Sigma_{t_0}}\bigg)^\frac{1}{2}\log\frac{t_0}{t}+|\partial^\beta_x\psi(t_0,x)|,
\end{align}
for any multi-index $\beta$, where $C$ is a constant independent of $t_0,p_i$.
\end{proposition}
\begin{proof}
We compute the divergence of the current $J^{e_0}_a[\psi]$:
\begin{align}\label{divJvarphi}
\notag\nabla^aJ^{te_0}_a[\psi]\overset{(\ref{divJa})}{=}&\,\nabla^a(te_0)^bT_{ab}[\psi]
=tK^{ab}T_{ab}[\psi]-(e_0t)T_{00}[\psi]\\
=&\sum_{i=1}^3K_{ii}t(e_i\psi)^2-\frac{1}{2}{K^i}_it|\nabla\psi|^2+\frac{1}{2}\big[(e_0\psi)^2+|\overline{\nabla}\psi|^2\big]\\
\notag\overset{(\ref{KasnerKii})}{=}&\,\sum_{i=1}^3(1-p_i)(e_i\psi)^2.
\end{align}
Hence, by (\ref{divthe2}), for $X=te_0$, we have
\begin{align}\label{Stokesvarphi}
t\int_{\Sigma_t}J^{e_0}_0[\psi]\mathrm{vol}_{\Sigma_t} =&\,t_0\int_{\Sigma_{t_0}}J^{e_0}_0[\psi]\mathrm{vol}_{\Sigma_{t_0}}+\int^{t_0}_t\int_{\Sigma_s} \sum_{i=1}^3 (p_i-1 )(e_i\psi)^2 )\mathrm{vol}_{\Sigma_s} ds\\
\tag{$p_i\leq1$}\leq&\,t_0\int_{\Sigma_{t_0}}J^{e_0}_0[\psi]\mathrm{vol}_{\Sigma_{t_0}}.
\end{align}
Note, that the same inequality holds for $\partial^\beta_x\psi$ by commuting the wave equation with $\partial^\beta_x$.
In particular, taking into account the volume form, we control
\begin{align}\label{dtvarphiest}
t^2\int_{\Sigma_t}(\partial_t\psi)^2\mathrm{vol}_{Euc}\leq 2t_0\int_{\Sigma_{t_0}}J^{e_0}_0[\psi]\mathrm{vol}_{\Sigma_{t_0}},&&t^2\|\partial_t\partial^\beta_x\psi\|_{H^2}^2\leq\sum_{|\alpha|\leq2} 2t_0\int_{\Sigma_{t_0}}J^{e_0}_0[\partial_{x}^\alpha\partial_x^\beta\psi]\mathrm{vol}_{\Sigma_{t_0}},
\end{align}
for all $t\in(0,t_0]$.

Using the fundamental theorem of calculus along $e_0$ and Sobolev embedding $H^2(\mathbb{T}^3)\hookrightarrow L^\infty(\mathbb{T}^3)$, we then derive
\begin{align}\label{Linftyvarphi}
|\partial^\beta_x\psi(t,x)|=&\, |\int^t_{t_0}\partial_s\partial^\beta_x\psi ds+\psi(t_0,x) |\\
\notag\leq&\int^{t_0}_tC\|\partial_s\partial^\beta_x\psi\|_{H^2(\mathrm{vol}_{Euc})}ds+|\partial^\beta_x\psi(t_0,x)|\\
\tag{by (\ref{dtvarphiest})}\leq&\,C\sum_{|\alpha|\leq 2}\bigg(t_0^2\int_{\Sigma_{t_0}}J^{e_0}_0[\partial_{x}^{\alpha}\partial^\beta_x\psi]\mathrm{vol}_{Euc}\bigg)^\frac{1}{2}\log\frac{t_0}{t}+|\partial^\beta_x\psi(t_0,x)|,
\end{align}
for all $t\in(0,t_0]$.
\end{proof}
Instead of deriving renormalised energy estimates, as in the previous subsection for FLRW (see Proposition \ref{prop:renestFLRW}), we prove the validity of the expansion (\ref{psiasymKasner}) by using (\ref{LinftypsiKasner}) to view the wave equation as an inhomogeneous ODE in $t$. This procedure is more wasteful in the number of derivatives of $\psi$ that we need to bound from initial data, but it is slightly simpler. 
\begin{proof}[Proof of Theorem \ref{Thm:Asym} for Kasner]
We express the wave equation for $\psi$ in terms of the $(t,x_1,x_2,x_3)$ coordinate system and treat the spatial derivatives of $\psi$ as error terms:
\begin{align}\label{boxpsiexp}
-\partial_t^2\psi-\frac{1}{t}\partial_t\psi+\sum_{i=1}^3t^{-2p_i}\partial_{x_i}^2\psi=0\qquad \Rightarrow\qquad \partial_t(t\partial_t\psi)=\sum_{i=1}^3t^{1-2p_i}\partial_{x_i}^2\psi.
\end{align}
Integrating in $[t,t_0]$ we obtain the formula
\begin{align}\label{boxpsiexp2}
\notag t\partial_t\psi=&\,t_0\partial_t\psi_0-\int^{t_0}_t\sum_{i=1}^3s^{1-2p_i}\partial_{x_i}^2\psi ds\\
\psi(t,x)=&\,\psi(t_0,x)+t_0\partial_t\psi_0\log\frac{t}{t_0}+\int^{t_0}_t\frac{1}{s}\int^{t_0}_s\sum_{i=1}^3\overline{s}^{1-2p_i}\partial_{x_i}^2\psi d\overline{s}ds\\
\notag=&\,\psi(t_0,x)+\bigg(t_0\partial_t\psi_0+\int^{t_0}_0\sum_{i=1}^3s^{1-2p_i}\partial_{x_i}^2\psi ds\bigg)\log\frac{t}{t_0}
+\int^{t_0}_t\frac{1}{s}\int^s_0\sum_{i=1}^3\overline{s}^{1-2p_i}\partial_{x_i}^2\psi d\overline{s} ds\\
\notag=&\,A(x)\log t+u(t,x),
\end{align}
where 
\begin{align}
\label{AKasner}A(x)=&\,t_0\partial_t\psi_0+\int^{t_0}_0\sum_{i=1}^3s^{1-2p_i}\partial_{x_i}^2\psi ds,\\
\label{uKasner}u(t,x)=&\,\psi(t_0,x)-\bigg(t_0\partial_t\psi_0+\int^{t_0}_0\sum_{i=1}^3s^{1-2p_i}\partial_{x_i}^2\psi ds\bigg)\log t_0+\int^{t_0}_t\frac{1}{s}\int^s_0\sum_{i=1}^3\overline{s}^{1-2p_i}\partial_{x_i}^2\psi d\overline{s} ds.
\end{align}
Note, that since by the assumption $1-2p_i>-1$ and by (\ref{LinftypsiKasner}) $\|\partial_x^\beta\psi\|_{L^\infty}\leq C|\log t|$, $t\in(0,t_0]$, the functions $s^{1-2p_i}\partial_{x_i}^2\psi$, $\frac{1}{s}\int^s_0\sum_{i=1}^3\overline{s}^{1-2p_i}\partial_{x_i}^2\psi d\overline{s}$ are integrable\footnote{By Proposition \ref{prop:upperestKasner}, their $L^1([0,t_0])$ norm is bounded by initial data.} in $[0,t_0]$ and hence the above formulas make sense. Moreover, it is implied by (\ref{AKasner}), (\ref{uKasner}) that $A(x),u(t,x)$ are smooth functions and $u=u_{Kasner}$ and its spatial derivatives are in fact uniformly bounded up to $t=0$:
\begin{align}\label{uasym}
\|\partial_x^\alpha u\|_{L^\infty(\Sigma_t)}\leq C,
\end{align}
for all $t\in(0,t_0]$, any multi-index $\alpha$, where $C>0$ is a constant depending on initial data.
\end{proof}
Next, we prove the blow up result stated in Theorem \ref{Thm:BlowUP} for Kasner. Notice that since $p_i<1$, the main contribution of the energy flux generated by the current  $J^{e_0}[\psi]$, as $t\rightarrow0$, comes from the $\partial_t\psi$ term:
\begin{align}\label{varphienbeh}
t\int_{\Sigma_t}J^{e_0}[\psi]\mathrm{vol}_{\Sigma_s}=&\,\frac{1}{2}\int_{\Sigma_t}t^2(\partial_t\psi)^2+\sum_{i=1}^3t^{2-2p_i}(\partial_{x_i}\psi)^2\mathrm{vol}_{Euc}\\
\notag=&\, \frac{1}{2}\int_{\Sigma_t}t^2(\partial_t\psi)^2\mathrm{vol}_{Euc}+\sum_{i=1}^3t^{2-2p_i}O(|\log t|^2).
\end{align}
Utilising (\ref{psiasymKasner}), (\ref{uasym}) it follows that
\begin{align}\label{varphienlim}
\lim_{t\rightarrow0}t\int_{\Sigma_t}J^{e_0}[\psi]\mathrm{vol}_{\Sigma_s}= \frac{1}{2}\lim_{t\rightarrow0}\int_{\Sigma_t}t^2(\partial_t\psi)^2\mathrm{vol}_{Euc}=\frac{1}{2}\int_{\Sigma_0}A^2(x)\mathrm{vol}_{Euc}.
\end{align}
Thus, returning to (\ref{Stokesvarphi}) and taking the limit $t\rightarrow0$ we obtain the identity:
\begin{align}\label{L2Aid}
\int_{\Sigma_0}A^2(x)\mathrm{vol}_{Euc}=&\,t_0\int_{\Sigma_{t_0}}J^{e_0}_0[\psi]\mathrm{vol}_{\Sigma_{t_0}}+\int^{t_0}_0\int_{\Sigma_s}\sum_{i=1}^3(p_i-1)(e_i\psi)^2\mathrm{vol}_{\Sigma_s} ds \\
\notag=&\,\frac{1}{2}t_0^2\|\partial_t\psi\|^2_{L^2(\Sigma_{t_0})}+\frac{1}{2}\sum_{i=1}^3t_0^{2-2p_i}\|\partial_{x_i}\psi\|_{L^2(\Sigma_{t_0})}^2\\
\notag&+\int^{t_0}_0\sum_{i=1}^3(p_i-1)s^{1-2p_i}\int_{\Sigma_s}(\partial_{x_i}\psi)^2\mathrm{vol}_{Euc} ds.
\end{align}
We bound the $L^2$ norm of $\partial_{x_i}\psi$ as follows:
\begin{lemma}\label{lem:L2estFLRW}
	The following estimate for the $L^2$ norm of $\partial_{x_i}\psi$ holds:
	\begin{align}\label{L2estFLRW}
	\|\partial_{x_i}\psi\|_{L^2(\Sigma_t)}\leq \|\partial_{x_i}\psi\|_{L^2(\Sigma_{t_0})}+\bigg(2t_0\int_{\Sigma_{t_0}}J^{e_0}_0[\partial_{x_i}\psi]\mathrm{vol}_{\Sigma_{t_0}}\bigg)^\frac{1}{2}\log\frac{t_0}{t},
	\end{align}
	for all $t\in(0,t_0]$.
\end{lemma}

\begin{proof}
We have
\begin{align}\label{L2logbound}
&\notag-\frac{1}{2}\partial_t \|\partial_{x_i}\psi\|^2_{L^2(\Sigma_t)}\overset{C-S}{\leq} \|\partial_{x_i}\psi\|_{L^2(\Sigma_t)}\|\partial_t\partial_{x_i}\psi\|_{L^2(\Sigma_t)}\\
& \notag-\partial_t \|\partial_{x_i}\psi\|_{L^2(\Sigma_t)}\leq \|\partial_t\partial_{x_i}\psi\|_{L^2(\Sigma_t)}\\
& \notag \|\partial_{x_i}\psi\|_{L^2(\Sigma_t)}\leq \|\partial_{x_i}\psi\|_{L^2(\Sigma_{t_0})}+\int^{t_0}_t \|\partial_s\partial_{x_i}\psi\|_{L^2(\Sigma_s)}ds\\
&\|\partial_{x_i}\psi\|_{L^2(\Sigma_t)}\leq \|\partial_{x_i}\psi\|_{L^2(\Sigma_{t_0})}
+\bigg(2t_0\int_{\Sigma_{t_0}}J^{e_0}_0[\partial_{x_i}\psi]\mathrm{vol}_{\Sigma_{t_0}}\bigg)^\frac{1}{2}\log\frac{t_0}{t},\qquad\text{by (\ref{dtvarphiest})}.
\end{align}	
\end{proof}

Applying (\ref{L2logbound}) to (\ref{L2Aid}) we derive:
\begin{align}\label{L2Alowbound}
&\int_{\Sigma_0}A^2(x)\mathrm{vol}_{Euc}\\
\notag\ge&\,
\frac{1}{2}t_0^2\|\partial_t\psi\|^2_{L^2(\Sigma_{t_0})}+\frac{1}{2}\sum_{i=1}^3t_0^{2-2p_i}\|\partial_{x_i}\psi\|_{L^2(\Sigma_{t_0})}^2
+\sum_{i=1}^3\int^{t_0}_0(p_i-1)s^{1-2p_i}ds(2\|\partial_{x_i}\psi\|_{L^2(\Sigma_{t_0})}^2)\\
\notag&+\sum_{i=1}^3\int^{t_0}_0 (p_i-1)s^{1-2p_i}|\log\frac{s}{t_0}|^2 ds \bigg(2t_0^2\|\partial_t\partial_{x_i}\psi\|_{L^2(\Sigma_{t_0})}+2\sum_{j=1}^3t_0^{2-2p_i}\|\partial_{x_j}\partial_{x_i}\psi\|_{L^2(\Sigma_{t_0})}\bigg)\\
\notag\ge&\,\frac{1}{2}t_0^2\|\partial_t\psi\|^2_{L^2(\Sigma_{t_0})}-\frac{1}{2}\sum_{i=1}^3t_0^{2-2p_i}\|\partial_{x_i}\psi\|_{L^2(\Sigma_{t_0})}^2\\
\notag
&-\sum_{i=1}^3\frac{t_0^{2-2p_i}}{(1-p_i)^2}\bigg[
t_0^2\|\partial_t\partial_{x_i}\psi\|^2_{L^2(\Sigma_{t_0})}+\sum_{j=1}^3t_0^{2-2p_i}\|\partial_{x_j}\partial_{x_i}\psi\|_{L^2(\Sigma_{t_0})}^2\bigg].
\end{align}
If the assumptions of Theorem \ref{Thm:BlowUP} for Kasner are satisfied, then it is clear from the preceding lower bound that $\|A(x)\|_{L^2(\mathbb{T}^3)}>0$.

To prove the local version of the blow up criterion in Theorem \ref{Thm:BlowUP_local}, we argue similarly, but also take into account the contribution of the null flux terms in (\ref{divthe2}). Note that the upper bounds (\ref{psienestKasner}), (\ref{LinftypsiKasner}), (\ref{L2logbound}) are also valid for the integral domains $U_t,U_0$ in place of $\Sigma_t,\Sigma_{t_0}$, since the null flux terms in (\ref{divthe2}) have a favourable sign for an upper bound. Hence, taking the limit $t\rightarrow0$ in (\ref{divthe}) and employing (\ref{n}), (\ref{divJvarphi}), (\ref{psiasymKasner}), (\ref{uasym}) we obtain:
\begin{align}\label{L2Alowboundlocal}
&\int_{U_0}A^2(x)\mathrm{vol}_{Euc}\\
\notag\ge&\,
\frac{1}{2}t_0^2\|\partial_t\psi\|^2_{L^2(U_{t_0})}+\frac{1}{2}\sum_{i=1}^3t_0^{2-2p_i}\|\partial_{x_i}\psi\|_{L^2(U_{t_0})}^2
+\int^{t_0}_0\sum_{i=1}^3(p_i-1)s^{1-2p_i}\int_{U_s}(\partial_{x_i}\psi)^2\mathrm{vol}_{Euc} ds\\
\tag{$i<j;\,i,j\neq l$}&-\sum_{l=1}^3\int^{t_0}_t\int_{\mathcal{N}_l^\pm\cap \overline{U}_s} s^{2-p_l}\frac{1}{2}\big[|(e_0\pm e_l)\psi|^2+|e_i\psi|^2+|e_j\psi|^2\big]dsdx_idx_j.
\end{align}
Since $t^{-p_l}dt=\pm dx_l$ along $\mathcal{N}_l^\pm$, it follows by integrating that the closure $\overline{U}_t$ of the neighbourhood of $U_t$ is the product $I_1\times I_2\times I_3$, where $I_i=[-\frac{t^{1-p_i}}{1-p_i},\delta+\frac{t^{1-p_i}}{1-p_i}]$. The analogous inequality to (\ref{Sob1D}) then reads
\begin{equation}\label{Sob1DKasner}
f^2(t,x_l)\leq (\delta +\frac{2}{1-p_l}t^{1-p_l})^{-1}\int_{I_l}f^2(t,x_l) dx_l+ \|f(t,x_l) \|^2_{H^1(I_l)},\qquad f\in H^1(I_l),\;\;t\in(0,t_0].
\end{equation}
Applying the latter bound to the integral over $\mathcal{N}_l^\pm\cap \overline{U}_s$ on the RHS of (\ref{L2Alowboundlocal}), along with (\ref{L2logbound}), we derive
\begin{align}\label{L2Alowboundlocal2}
&\int_{U_0}A^2(x)\mathrm{vol}_{Euc}\\
\notag\ge&\,
\frac{1}{2}t_0^2\|\partial_t\psi\|^2_{L^2(U_{t_0})}+\frac{1}{2}\sum_{i=1}^3t_0^{2-2p_i}\|\partial_{x_i}\psi\|_{L^2(U_{t_0})}^2
+\sum_{i=1}^3\int^{t_0}_0(p_i-1)s^{1-2p_i}ds(2\|\partial_{x_i}\psi\|_{L^2(U_{t_0})}^2)\\
\notag&+\sum_{i=1}^3\int^{t_0}_0 (p_i-1)s^{1-2p_i}|\log\frac{s}{t_0}|^2 ds \bigg(2t_0^2\|\partial_t\partial_{x_i}\psi\|_{L^2(U_{t_0})}+2\sum_{j=1}^3t_0^{2-2p_i}\|\partial_{x_j}\partial_{x_i}\psi\|_{L^2(U_{t_0})}\bigg)\\
\notag&-\sum_{l=1}^3\bigg[\int^{t_0}_t(1+\frac{1}{\delta +\frac{2}{1-p_l}s^{1-p_l}})s^{-p_l}\int_{U_s}4sJ^{e_0}_0[\psi]\mathrm{vol}_{U_s}ds
+\int^{t_0}_ts^{-p_l}\int_{U_s}4sJ^{e_0}_0[\partial_{x_l}\psi]\mathrm{vol}_{U_s}ds\bigg]\\
\notag\ge&\,\frac{1}{2}t_0^2\|\partial_t\psi\|^2_{L^2(U_{t_0})}-\frac{1}{2}\sum_{i=1}^3t_0^{2-2p_i}\|\partial_{x_i}\psi\|_{L^2(U_{t_0})}^2\\
\notag&-\sum_{i=1}^3\frac{t_0^{2-2p_i}}{(1-p_i)^2}\bigg[
t_0^2\|\partial_t\partial_{x_i}\psi\|^2_{L^2(\Sigma_{t_0})}+\sum_{j=1}^3t_0^{2-2p_i}\|\partial_j\partial_{x_i}\psi\|_{L^2(\Sigma_{t_0})}^2\bigg]\\
\notag&-\sum_{l=1}^3\frac{2t_0^{1-p_l}}{1-p_l}\bigg[
t^2_0\|\partial_t\psi\|^2_{L^2(U_{t_0})}+\sum_{i=1}^3t^{2-2p_i}_0\|\partial_{x_i}\psi\|^2_{L^2(U_{t_0})}\\
\notag&+t^2_0\sum_{l=1}^3\|\partial_t\partial_{x_l}\psi\|^2_{L^2(U_{t_0})}+\sum_{i=1}^3t^{2-2p_i}_0\|\partial_{x_i}\partial_{x_l}\psi\|^2_{L^2(U_{t_0})}
\bigg]\\
\notag&-\sum_{l=1}^3\log\big(1+\frac{2}{1-p_l}\frac{t^{1-p_l}_0}{\delta}\big)\big[
t^2_0\|\partial_t\psi\|^2_{L^2(U_{t_0})}+t^{2-2p_l}_0\sum_{i=1}^3\|\partial_{x_i}\psi\|^2_{L^2(U_{t_0})}
\big].
\end{align}
Thus, given the assumptions in Theorem \ref{Thm:BlowUP_local} for Kasner, it follows that $\|A(x)\|_{L^2(U_0)}>0$, as required.

\section*{Acknowledgements}
The authors would like to thank Mihalis Dafermos for useful interactions. Further, A.A. and A.F. benefited from discussions with Jos\'e Nat\'ario.
This work was partially supported by FCT/Portugal through UID/MAT/04459/2013, grant (GPSEinstein) PTDC/MAT-ANA/1275/2014 and through the FCT fellowships SFRH/BPD/115959/2016 (A.F.) and SFRH/BPD/85194/2012 (A.A.). G.F. is supported by the EPSRC grant EP/K00865X/1 on ‘Singularities of Geometric Partial Differential Equations’.

\end{document}